\documentclass[11pt,a4paper]{article}
\usepackage[T1]{fontenc}
\usepackage{lmodern}
\usepackage{amsmath,amssymb,amsthm}
\usepackage{fullpage}
\usepackage{tabularx}
\usepackage[numbers,sort&compress]{natbib}
\usepackage[hidelinks]{hyperref}
\usepackage{doi}
\usepackage[capitalise,noabbrev]{cleveref}
\newtheorem{theorem}{Theorem}
\newtheorem{lemma}{Lemma}
\DeclareMathOperator{\conv}{conv}
\DeclareMathOperator{\enc}{enc}
\newcommand{\R}{\mathbb R}
\newcommand{\Z}{\mathbb Z}
\newcommand{\boxproblem}[4]{%
  \begin{center}
  \fbox{\begin{minipage}{0.97\linewidth}
    \vspace{2pt}\noindent\normalsize\textsc{#1}\par\vspace{2pt}
    \setlength{\tabcolsep}{3pt}\renewcommand{\arraystretch}{1.0}
    \begin{tabularx}{\linewidth}{@{}lX@{}}
      \textbf{Input:}& #2\\
      \textbf{Question:}& #3\\
      \textbf{Parameter:}& #4
    \end{tabularx}
  \end{minipage}}
  \end{center}}
\title{High-Multiplicity Bin Packing is FPT}
\author{Tomohiro Koana\thanks{Graduate School of Information Science and Technology, The University of Tokyo, Japan.\newline Email: \texttt{tomohiro.koana@gmail.com}.}
\and Soh Kumabe\thanks{CyberAgent, Inc., Tokyo, Japan. Email: \texttt{kumabe\_soh@cyberagent.co.jp}.}}
\date{}
\begin{document}
\maketitle

\begin{abstract}
Bin packing asks whether a collection of items can be packed into at most a
given number of bins of a given capacity.
We consider the high-multiplicity setting with $d$ distinct item sizes, in
which both the item sizes and the number of items of each size are encoded
in binary.
\citeauthor{GR20} (JACM~\citeyear{GR20}) gave an XP algorithm parameterized by
$d$.
Whether this problem is fixed-parameter tractable (FPT) in $d$ has remained a
central open problem.

We resolve this question by giving a deterministic $O^*(2^{d^{O(d)}})$-time
algorithm.
We formulate bin packing as an integer linear program (ILP) with at most
$(d+1)d^d$ variables.
A bin configuration records the number of items of each type in one bin.
We partition these configurations by their coordinate remainders modulo $d$.
For each class, we use one variable for the bin count and $d$ variables for
the total item counts.
The convex hull of each class has the integer decomposition property,
which guarantees that every feasible ILP solution corresponds to a packing.
\end{abstract}

\section{Introduction}\label{sec:introduction}
Bin packing is a classical problem in combinatorial optimization: given item
sizes and a bin capacity, the goal is to pack all items into as few bins as
possible.
NP-hardness already follows from the two-bin case, which contains
\textsc{Partition}, one of the 21 problems proved NP-complete by Karp~\cite{Karp72}.
We consider the high-multiplicity version, in which each item size is given
together with its multiplicity.
\boxproblem{Bin Packing}%
{A capacity $B\in\Z_{\ge1}$, $d\ge1$ item sizes $s_1,\ldots,s_d\in\Z_{\ge1}$,
 multiplicities $a_1,\ldots,a_d\in\Z_{\ge0}$, and a bin bound $b\in\Z_{\ge0}$.}%
{Can $a_i$ items of size $s_i$ for each $1\le i\le d$ be packed into at most
 $b$ bins of capacity $B$?}%
{$d$}
The question of whether \textsc{Bin Packing} is fixed-parameter tractable (FPT)
in $d$ has been asked repeatedly in the
literature~\cite{GR20,MB18,KLMPS24,KZ25,JKZ25,JOP26}.

Bin packing with few item types has been studied extensively.
McCormick, Smallwood, and Spieksma~\cite{MSS01} gave a polynomial-time
algorithm for $d=2$ in 2001.
Goemans and Rothvoss~\cite{GR20} subsequently established polynomial-time
solvability for every fixed $d$ in 2014.
Their algorithm runs in $L^{2^{O(d)}}$ time, where $L$ denotes the binary
encoding length of the input.
The exponent in $L^{2^{O(d)}}$ depends on $d$, so the algorithm is XP.

Jansen and Solis-Oba~\cite{JSO11} gave an additive-one approximation that uses
at most one bin more than the optimum and runs in $2^{2^{O(d)}}L^{O(1)}$ time.
Jansen and Klein~\cite{JK20} gave an algorithm running in
$|V|^{2^{O(d)}}L^{O(1)}$ time, where $V$ is the vertex set of the convex hull
of feasible vectors that record how many items of each type are packed in one bin.
However, neither result shows that \textsc{Bin Packing} is FPT in $d$ alone:
the approximation permits an extra bin, while $|V|$ is not bounded by a
function of $d$ alone.

Variants in which either the item sizes or the multiplicities are encoded in
unary have also been studied.
For unary-encoded sizes $s_i$ and binary-encoded multiplicities $a_i$, the
Goemans--Rothvoss framework~\cite{GR20}, with the refinements of Jansen, Kahler,
and Zwanger~\cite{JKZ25}, yields an FPT algorithm in $d$ with running time
$2^{2^{O(d)}}$ times a polynomial in the input length.
For unary-encoded multiplicities $a_i$ and binary-encoded sizes $s_i$,
Kouteck\'y and Zink~\cite{KZ25} gave an FPT algorithm in $d$ with running time
$2^{2^{O(d)}}$ times a polynomial in the input length.
However, when both the item sizes $s_i$ and the multiplicities $a_i$ are
simultaneously encoded in binary, FPT in $d$ alone remained open.

Under the Exponential Time Hypothesis (ETH), Kowalik et al.~\cite{KLMPS24}
ruled out running time $N^{2^{o(q)}}$ for \textsc{Point in Cone} in ambient
dimension $q$, where $N$ is the binary input length.
This problem asks whether a given integer point is a nonnegative integer
combination of integer points in a given bounded rational polytope.
Jansen, Ohnesorge, and Pirotton~\cite{JOP26} recently proved a lower bound
specifically for high-multiplicity bin packing: ETH rules out running time
$L^{2^{o(d)}}$.
Neither lower bound rules out FPT, but the latter implies that under ETH any
FPT running time $f(d)L^{O(1)}$ must have $f(d)=2^{2^{\Omega(d)}}$.

\paragraph{Our contribution.}
We resolve this long-standing question: \textsc{Bin Packing} admits an FPT algorithm.\footnote{Running times
count bit operations; $O^*(\cdot)$ suppresses factors polynomial in the binary
input length, with exponent independent of $d$.}
\begin{theorem}\label{thm:main}
\textnormal{\textsc{Bin Packing}} can be decided deterministically in
$O^*(2^{d^{O(d)}})$ time.
\end{theorem}

We first express \textsc{Bin Packing} as an integer linear program (ILP) with
at most $(d+1)d^d$ variables and then solve the ILP using separation oracles
and Lenstra's algorithm~\cite{GLS88,Len83,DV12}.

To obtain the ILP, we represent the contents of one bin by a \emph{bin configuration},
a vector recording the number of items of each type.
We partition the bin configurations according to their coordinate remainders
modulo $d$.
For each class $r$, we use one integer variable for the bin count and
$d$ integer variables for total item counts.
The bin average configuration must lie in the convex
hull $P_r$ of the configurations in $r$.
This gives our ILP formulation.
To prove correctness, we establish that each polytope $P_r$ has the
\emph{integer decomposition property}.
This guarantees that any integral vector of item totals in a copy of $P_r$
scaled by the bin count can be written as a sum of the specified number of
feasible bin configurations.

\paragraph{Related work.}
The \emph{support} of a packing is the number of distinct vectors of item counts
used by its bins.
Eisenbrand and Shmonin~\cite{ES06} proved Carath\'eodory-type bounds for integer
cones, implying that every feasible bin packing instance has an optimal packing
with support at most $2^d$.
Jansen, Pirotton, and Tutas~\cite{JPT25} constructed instances with $d$ item types
for which every optimal packing has support $2^{\Omega(d)}$, showing that an
exponential dependence on $d$ is necessary in the support bound.

Approximation algorithms for high-multiplicity bin packing provide additive
guarantees in terms of $d$.
Let $\mathrm{OPT}$ denote the minimum number of bins.
LP rounding of the Gilmore--Gomory formulation~\cite{GG61} gives a packing
using at most $\mathrm{OPT}+d$ bins.
Filippi and Agnetis~\cite{FA05} gave a polynomial-time algorithm for every
fixed $d\ge2$ using at most $\mathrm{OPT}+d-2$ bins.
Filippi~\cite{F07} gave a polynomial-time algorithm for every fixed $d$ with
additive error one for $3\le d\le6$ and $1+\lfloor(d-1)/3\rfloor$ for $d>6$.
Jansen and Solis-Oba~\cite{JSO11} obtained a packing using at most
$\mathrm{OPT}+1$ bins in FPT time parameterized by $d$.

FPT algorithms for bin packing under other parameterizations have also
been studied.
With unary-encoded multiplicities $a_i$ and binary-encoded sizes $s_i$,
Mnich and Wiese~\cite{MW15} obtained FPT in the largest integer item size
$\max_i s_i$.
For explicitly listed items, Bannach et al.~\cite{BBMMLRS20} gave a deterministic
FPT algorithm parameterized by the number of items of size at most $B/3$.
For high-multiplicity vector bin packing, including bin packing with
cardinality constraints, Knop et al.~\cite{KKLMO21} obtained FPT
jointly parameterized by the number of item types, the vector dimension, and
the largest integer size coordinate.
In contrast, Jansen, Kratsch, Marx, and Schlotter~\cite{JKMS13} showed that bin
packing is W[1]-hard parameterized by the number of bins $b$, even when both
$s_i$ and $a_i$ are encoded in unary.

\paragraph{Organization.}
The rest of this paper is organized as follows.
\Cref{sec:formulation} gives an ILP formulation of \textsc{Bin Packing}.
\Cref{sec:oracles} constructs the separation oracle and proves \cref{thm:main}.

\section{An integer linear programming formulation}\label{sec:formulation}\label{sec:cover}
In this section, we formulate \textsc{Bin Packing} as an ILP with $(d+1)d^d$
variables.
Write $s:=(s_1,\ldots,s_d)$ and $a:=(a_1,\ldots,a_d)$.
A \emph{bin configuration} is a vector whose $i$th coordinate records the number
of items of type $i$ packed in one bin.
To map a packing to an ILP solution, we group its bins by the coordinate
remainders of their configurations modulo $d$.
All residue indices $r$ range over $\{0,\ldots,d-1\}^d$.
For each $r$, we use integer variables $n_r$ and $z_r\in\Z^d$, whose intended meanings are 
the number of bins in that group and the sum of their configurations,
respectively.
We impose the constraint $(z_r,n_r)\in C_r$, where $C_r$ is a cone defined as
follows.

Define the set $S$ of bin configurations and its convex hull $P$ by
\[
  S:=\{x\in\Z_{\ge0}^d:s^\top x\le B\},\qquad P:=\conv(S).
\]
For every $r$, let $S_r$ consist of the bin configurations congruent to $r$
coordinatewise modulo $d$, and let $P_r$ be the convex hull of $S_r$:
\[
  S_r:=S\cap(r+d\Z^d),\qquad P_r:=\conv(S_r).
\]
Since the nonnegativity and capacity inequalities hold throughout $P$,
we have $P\cap\Z^d=S$.
Thus $P_r\subseteq P$ and $S=\bigcup_r(P_r\cap\Z^d)$.
Although integer points of $P_r$ need not be congruent to $r$ modulo $d$,
they are all feasible bin configurations: $P_r\cap\Z^d\subseteq S$.

For every $r$, define the rational polyhedral cone $C_r\subseteq\R^{d+1}$ by
\[
  C_r:=\{(0,0)\}\cup\{(np,n):p\in P_r,\ n\in\R_{>0}\}.
\]

Let $D:=(d+1)d^d$.
Let $K\subseteq\R^D$ be the rational polytope defined by the following constraints,
imposed for every $r$ and $1\le i\le d$:
\[
  \sum_r z_r=a,\qquad \sum_r n_r=b,\qquad
  0\le n_r\le b,\qquad 0\le z_{r,i}\le a_i,\qquad (z_r,n_r)\in C_r.
\]
We call the equalities and variable bounds the \emph{count constraints}, and the
requirements $(z_r,n_r)\in C_r$ the \emph{cone constraints}.
Our ILP asks whether $K\cap\Z^D$ is nonempty.

The rest of this section is devoted to proving the correctness of this ILP
formulation.
The forward direction is straightforward (\cref{lem:packing-to-ilp}), while the
converse uses the integer decomposition property to obtain $n_r$ bin configurations whose sum is $z_r$
from any integral point $(z_r,n_r)\in C_r$.

\begin{lemma}\label{lem:packing-to-ilp}
If the \textnormal{\textsc{Bin Packing}} instance is feasible, then
$K\cap\Z^D\ne\emptyset$.
\end{lemma}
\begin{proof}
Let $x^{(1)},\ldots,x^{(b)}\in S$ be the bin configurations of a feasible packing,
with empty bins added if necessary.
For each $r$, define
$J_r:=\{j\in\{1,\ldots,b\}:x^{(j)}\in S_r\}$,
$n_r:=|J_r|$, and $z_r:=\sum_{j\in J_r}x^{(j)}$.
Since the sets $J_r$ partition $\{1,\ldots,b\}$ and all $x^{(j)}$ are nonnegative,
the count constraints hold.
If $n_r>0$, then $z_r/n_r\in P_r$ by convexity, while $n_r=0$ gives $z_r=0$.
Thus the cone constraints also hold, and $((z_r,n_r))_r$ is an integral point of $K$.
\end{proof}

We now prove the converse direction.
A \emph{lattice polytope} in $\R^d$ is the convex hull of a finite nonempty subset
of $\Z^d$.
For a polytope $Q\subseteq\R^d$ and a real number $n>0$, we define
$nQ:=\{nx:x\in Q\}$.
A lattice polytope $Q\subseteq\R^d$ has the \emph{integer decomposition property}
if, for every integer $n\ge1$, each $z\in nQ\cap\Z^d$ can be expressed as a sum
of $n$ points of $Q\cap\Z^d$.

We establish the integer decomposition property of every nonempty $P_r$
using the following standard lemma
(see \cite[Theorem~2.2.12]{CLS11}).
Once this property is established, the converse direction follows immediately
as well (\cref{lem:ilp-to-packing}).
\begin{lemma}\label{lem:dilation}
Let $V\subseteq\Z^d$ be finite and nonempty, let $k\ge d$ be an integer, and
let $r\in\Z^d$.
Then $r+k\conv(V)$ has the integer decomposition property.
\end{lemma}

We include a proof for completeness, using the following form of
Carath\'eodory's theorem.
\begin{lemma}[Carath\'eodory's theorem; \cite{Rock70}, Theorem~17.1]\label{lem:caratheodory}
Let $V\subseteq\R^d$ be finite and nonempty, let $t>0$, and let $z\in t\conv(V)$.
There exist points $v_1,\ldots,v_{d+1}\in V$ and nonnegative real numbers
$\alpha_1,\ldots,\alpha_{d+1}$ such that
$z=\sum_{i=1}^{d+1}\alpha_i v_i$ and $\sum_{i=1}^{d+1}\alpha_i=t$.
\end{lemma}

\begin{proof}[Proof of \cref{lem:dilation}]
First consider $r=0$.
Fix an integer $n\ge1$ and a point $z\in nk\conv(V)\cap\Z^d$.
By \cref{lem:caratheodory}, choose points $v_1,\ldots,v_{d+1}\in V$ and nonnegative
coefficients $\alpha_1,\ldots,\alpha_{d+1}$ with
$z=\sum_{i=1}^{d+1}\alpha_i v_i$ and $\sum_{i=1}^{d+1}\alpha_i=nk$.
Since $\sum_{i=1}^{d+1}\lfloor\alpha_i\rfloor$ is an integer exceeding $nk-d-1$,
\[
  \sum_{i=1}^{d+1}\lfloor\alpha_i\rfloor\ge nk-d\ge(n-1)k.
\]
Thus we can take $(n-1)k$ points from the collection containing
$\lfloor\alpha_i\rfloor$ copies of $v_i$ for each $i$.
Partition the $(n-1)k$ points into $n-1$ groups of $k$ points and let
$z_2,\ldots,z_n$ be their sums.
For every $2\le j\le n$, the vector $z_j$ is integral and $z_j/k\in\conv(V)$.
Define the vector $z_1:=z-\sum_{j=2}^n z_j$.
Each coefficient $\alpha_i$ is decreased by at most $\lfloor\alpha_i\rfloor$, so
$z_1$ is a nonnegative linear combination of $v_1,\ldots,v_{d+1}$ with coefficient
sum $nk-(n-1)k=k$.
Thus $z_1\in k\conv(V)\cap\Z^d$ and $z=\sum_{j=1}^n z_j$.

Now take $z\in n(r+k\conv(V))\cap\Z^d$.
Since $z-nr\in nk\conv(V)\cap\Z^d$, the case $r=0$ gives
$z-nr=\sum_{j=1}^n z_j$ with $z_j\in k\conv(V)\cap\Z^d$.
Hence $z=\sum_{j=1}^n(r+z_j)$, and every $r+z_j$ belongs to
$(r+k\conv(V))\cap\Z^d$.
\end{proof}

We now use \cref{lem:dilation} to obtain a feasible packing from a feasible ILP solution.
\begin{lemma}\label{lem:ilp-to-packing}
If $K\cap\Z^D\ne\emptyset$, then the
\textnormal{\textsc{Bin Packing}} instance is feasible.
\end{lemma}
\begin{proof}
Take $((z_r,n_r))_r\in K\cap\Z^D$.
The cone constraints give $z_r=0$ when $n_r=0$ and
$z_r\in n_rP_r\cap\Z^d$ when $n_r>0$.
In the latter case, $P_r=r+d\conv((S_r-r)/d)$, where $(S_r-r)/d$ is a finite
nonempty subset of $\Z^d$.
By \cref{lem:dilation}, $P_r$ has the integer decomposition property, so $z_r$
is a sum of $n_r$ configurations in $P_r\cap\Z^d\subseteq S$.
Using these configurations as bins gives a feasible packing, since the count
constraints ensure $\sum_r n_r=b$ and $\sum_r z_r=a$.
\end{proof}

By \cref{lem:packing-to-ilp,lem:ilp-to-packing}, \textsc{Bin Packing} is equivalent
to an ILP with $D=(d+1)d^d$ variables.
However, the cone constraints specify membership in $C_r$, rather than an explicit
list of linear inequalities.
Without an efficiently computable inequality description of $K$, we cannot
directly apply an algorithm for explicitly given ILPs such as
Lenstra's~\cite{Len83}.
We instead construct a strong separation oracle for $K$ (\cref{sec:oracles}).

\section{Separation and the algorithm}\label{sec:oracles}\label{sec:algorithm}
We prove \cref{thm:main} using integer programming over
bounded convex sets.
A \emph{strong separation oracle} for a closed convex set $Q\subseteq\R^n$,
on input $y\in\mathbb Q^n$, either certifies $y\in Q$ or returns
$u\in\mathbb Q^n$ and $\alpha\in\mathbb Q$ such that
$u^\top x\le\alpha<u^\top y$ for every $x\in Q$.
For a rational number or vector $q$, let $\enc(q)$ denote its total binary
encoding length, including numerators and denominators.
In the bounds for separation algorithms, $\ell$ denotes the encoding length
of the query vector.

We use the following specialization of \cite[Theorem~7.1.1]{Dad12} to the
lattice $\Z^n$, with the deterministic algorithm of
\cite[Theorem~5.4]{DV12}.
\begin{lemma}[\cite{Dad12,DV12}]\label{lem:oracle-ip}
Let $Q\subseteq[-T,T]^n$ be a closed convex set, possibly empty, where $n\ge1$
and the rational number $T\ge1$ is given.
Let $\sigma\ge n+\enc(T)$ be a given integer.
Suppose strong separation over $Q$ returns inequalities of encoding length
$(\sigma+\ell)^{O(1)}$.
One can deterministically find a point of $Q\cap\Z^n$ or certify that this
intersection is empty using at most $2^{O(n^3)}\sigma^{O(1)}$ calls to strong
separation and additional bit operations, with query lengths at most
$2^{O(n^3)}\sigma^{O(1)}$.
\end{lemma}

Let $T:=1+\max\{b,a_1,\ldots,a_d\}$.
The count constraints give $K\subseteq[-T,T]^D$ and $\enc(T)=O(L)$.
By \cref{lem:oracle-ip}, with $n=D$ and
$\sigma=D+L+\enc(T)$, it remains to construct strong separation over $K$
with output length $(D+L+\ell)^{O(1)}$.
We first obtain separation over each $P_r$ from optimization, and then
combine these oracles with the count constraints.

A \emph{strong optimization oracle} for a nonempty rational polytope
$Q\subseteq\R^n$, on input $u\in\mathbb Q^n$, returns a rational point of $Q$
maximizing $u^\top x$.
The following result gives the required conversion.
\begin{lemma}[\cite{GLS88}, Theorem~6.4.9]\label{lem:optimization-separation}
Let $Q\subseteq\R^n$ be a nonempty rational polytope with vertex encoding
length at most $v$.
Suppose strong optimization over $Q$, on input $u\in\mathbb Q^n$, returns
a point of encoding length $(n+v+\enc(u))^{O(1)}$.
Strong separation over $Q$ can be implemented using $(n+v+\ell)^{O(1)}$
calls to strong optimization and additional bit operations, with query and
output lengths $(n+v+\ell)^{O(1)}$.
\end{lemma}

We implement optimization over $P_r$ by an ILP with $d$ variables.
\begin{lemma}\label{lem:piece-separation}
For every $r$, strong separation over $P_r$ takes
$2^{O(d^3)}(L+\ell)^{O(1)}$ bit operations, with output length
$(d+L+\ell)^{O(1)}$.
\end{lemma}
\begin{proof}
If $s^\top r>B$, then $P_r=\emptyset$, so return $0^\top x\le-1$.
Otherwise, $r\in S_r$, and for every $u\in\mathbb Q^d$ we have
\[
  \max_{x\in P_r}u^\top x
  =\max\{u^\top(r+dy):y\in\Z_{\ge0}^d,\ s^\top(r+dy)\le B\}.
\]
After clearing denominators, the right-hand side is a bounded integer program
with $d$ variables and encoding length polynomial in $d+L+\enc(u)$.
Lenstra's algorithm~\cite{Len83} returns an optimal integer vector $y^*$ in
$2^{O(d^3)}(L+\enc(u))^{O(1)}$ bit operations.
Returning $r+dy^*$ therefore solves strong optimization over $P_r$.

Every point of $S_r$, and hence every vertex of $P_r$ and the returned point,
has encoding length $O(d\log(B+1))$, since $S_r\subseteq[0,B]^d$.
Applying \cref{lem:optimization-separation} gives the claimed
separation time and output-length bound.
\end{proof}

We now combine the oracles for $P_r$ into an oracle for $K$.
\begin{lemma}\label{lem:aggregate-separation}
Strong separation over $K$ takes $2^{O(d^3)}(L+\ell)^{O(1)}$ bit operations,
with output length $(D+L+\ell)^{O(1)}$.
\end{lemma}
\begin{proof}
Given a rational vector $((z_r,n_r))_r$, first check the count constraints
and return any violated constraint, orienting an equality as a separating
inequality.
We may now assume that $n_r\ge0$ and $z_r\ge0$ for every $r$.

If $n_r=0$, the cone constraint holds exactly when $z_r=0$.
For $z_r\ne0$, return $s^\top z_r-Bn_r\le0$.
This inequality is valid on $C_r$ because every configuration has total size
at most $B$, and it is violated since $z_r\ge0$, $z_r\ne0$, and $s_i\ge1$.
If $n_r>0$, use \cref{lem:piece-separation} to test $z_r/n_r\in P_r$.
Given a separating inequality $u^\top x\le\alpha$, return
$u^\top z_r-\alpha n_r\le0$.
This inequality is valid on $C_r$ and violated by the query since $n_r>0$.
If all checks pass, every cone constraint holds, so the query belongs to $K$.

There are at most $d^d$ calls to \cref{lem:piece-separation}, each on
a vector of encoding length $\ell^{O(1)}$.
Together with $(D+L+\ell)^{O(1)}$ additional bit operations and
$D=(d+1)d^d$, this gives the claimed running time.
Padding a returned inequality with zeros outside its block gives output
length $(D+L+\ell)^{O(1)}$.
\end{proof}

\begin{proof}[Proof of \cref{thm:main}]
Apply \cref{lem:oracle-ip} to $K$ with $n=D$,
$T=1+\max\{b,a_1,\ldots,a_d\}$, and $\sigma=D+L+\enc(T)$, using
\cref{lem:aggregate-separation}.
This decides whether $K\cap\Z^D$ is nonempty and hence, by
\cref{lem:packing-to-ilp,lem:ilp-to-packing}, decides \textsc{Bin Packing}.
The oracle call and query-length bounds, together with the separation cost,
give $2^{O(D^3)}(D+L)^{O(1)}$ bit operations.
Since $D=(d+1)d^d$, the running time is $O^*(2^{d^{O(d)}})$.
\end{proof}

\section*{Acknowledgements}
Tomohiro Koana was supported in part by JST CREST Grant Number JPMJCR24Q2 and JST ERATO Grant Number JPMJER2301.

\section*{Declaration of generative AI use}
ChatGPT 6 Astra generated the proof of \cref{thm:main} and was also used to
draft the manuscript.
The authors verified and revised the proof and the manuscript and take full
responsibility for the paper.

\bibliographystyle{binpacking}
\begingroup
\raggedright
\bibliography{references}
\endgroup
\end{document}